\documentclass[preprint,12pt, unsort]{elsarticle}

\usepackage{amssymb}
\usepackage{graphicx}
\usepackage{graphicx}
\usepackage{times}  
\usepackage{helvet}  
\usepackage{courier}  
\usepackage{url}  
\usepackage{amsfonts}
\usepackage{amssymb}
\usepackage{amsmath}
\usepackage{algorithm}
\usepackage{algorithmic}
\usepackage{multirow}
\usepackage{color}
\usepackage{bbding}
\usepackage{appendix}
\usepackage{amsthm}

\newtheorem{lem}{Lemma}
\newtheorem{thm}{Theorem}
\newtheorem{defi}{Definition}

\newcommand{\be}{\begin{equation}}
\newcommand{\ee}{\end{equation}}
\def\bea{\begin{eqnarray}}
\def\eea{\end{eqnarray}}
\def\ba{\begin{array}}
\def\ea{\end{array}}
\def\nn{\nonumber}

\begin{document}

\begin{frontmatter}

\title{ Streaming algorithm for balance gain and cost with cardinality constraint on the integer lattice\tnoteref{t1}}
\author[1]{Jingjing Tan }
\ead{tanjingjing1108@163.com}

\cortext[cor1]{Corresponding author}
\address[1]{School of Mathematics and Statistics, Weifang University,\\
Weifang 261061,  P.R. China}

\begin{abstract}
  Team formation problem is a very important problem in the labor market, and it is proved to be NP-hard. In this paper, we design an efficient bicriteria streaming algorithms to construct a balance between gain and cost in a team formation problem with cardinality constraint on the integer lattice.
  To solve this problem, we establish a model for maximizing the difference between a nonnegative normalized monotone submodule function and a nonnegative linear function.
  Further, we discuss the case where the first function is $\alpha$--weakly submodular.
  Combining the lattice binary search with the threshold method, we present an online algorithm called  bicriteria streaming algorithms. Meanwhile,  we give detailed analysis for both of these models.

\end{abstract}

\begin{keyword}
Integer lattice \sep  Streaming algorithm \sep Knapsack constraint  \sep Cardinality function
\end{keyword}

\end{frontmatter}

\section{ Introduction}
Team formation problems arise in various contexts, such as project management, workforce allocation, sports team composition, and academic collaboration, and have recently attracted much attention\cite{2012Ana, 2018Ana,2009Lappas,2022Wang}. The goal of the team formation problem is a balance between the benefits gained and the costs incurred at the same time, which varies widely depending on the context. For example, in project management, the objective might be to maximize team productivity or minimize project completion time. In sports team composition, the goal may be to assemble a team with complementary skills and abilities to enhance overall performance. These problems can be formulated as a submodular optimization problem, aiming to select a set of elements $ V $ from a collection $ E $ in order to maximize a certain objective function $ g(V) $. In many practical scenarios, it is necessary to strike a balance between the benefits of selecting $ V $, quantified by the function $ g(V) $, and the associated cost denoted as $ c(V) $. To cater to these application domains, Nikolakaki et al.\cite{2020Ene} characterized the team formation problem as the following maximization problem model: $$ \max\limits_{ V \in E }\  g(V)-c(V), $$
where the function $ g(V) $ is monotone and non-negative submodular while $ c(V) $ represents a non-negative linear function that sums up the costs of selected elements. This model provides generalization for various data mining applications.
When the maximum number of elements to be selected is given as part of the input using the unconstrained model. Consider a constrained model when it comes to finding the optimal number of elements to be part of the solution. For cardinality constraints, they used the standard greedy algorithm to design algorithms that have provable approximation guarantees, where the elements of the collection are parted into groups. In designing cardinal-constraint algorithms with provable approximate guarantees, they taked the standard greedy algorithm where the elements of the collection are parted into groups. This method is very effective in solving radix constraint problems and can provide better approximate solutions. By grouping elements, it can better handle situations of higher complexity and have more flexibility to deal with various constraints, such as matroid constraints. In addition, they also consider online versions of the above questions. By using a variant of the continuous greedy technique, Sviridenko et al. \cite{2017Sviridenko}  acquired a solution set S such that
$ g(V)-c(V) \geq ( 1 - e^{-1} ) \cdot g(V) - c(V) $
for the same model with a matroid constraint.
Du et al.\cite{2014Du} proposed a two-criterion approximation algorithm, which is considered to be a more suitable performance measurement method compared with traditional approximation algorithms.
Feldman \cite{2019Feldman} designed bicriteria approximation algorithm and obtained the same conclusion as that in Sviridenko et al. \cite{2017Sviridenko}.

In the team formation problem, we will also encounter a fixed skill can be owned by multiple team members at the same time, in this case, we need to decide not only which skill to choose, but also to decide the number of choices, then the problem becomes a integer lattice submodule maximization problem. The concepts of lattice theory and optimization was first proposed by Lovasz \cite{1979 Lovasz}. The modular maximization problem of integer lattices is a classical extension of the modular maximization problem. For a finite set $E$ with size $n$, for each ${e_i}\in E$, let $\chi_{e_i}$ be the $i$-th unit vector. Denote by $\mathbb{N}^{E}$  the integer lattices defined on the $E$, for any $\boldsymbol{x}, \boldsymbol{y}$ with $\boldsymbol{x} \leq \boldsymbol{y}$, if the following inequality
 $$     g( \boldsymbol{y}+\chi_{e_i} ) - g( \boldsymbol{y} )
   \leq f( \boldsymbol{x}+\chi_{e_i} ) - f( \boldsymbol{x} ) $$
holds, we say that $ g $ satisfies the  DR-submodularity and that $ g $ is called a DR-submodular  function.
Soma et al. \cite{2018Soma} Studied DR-submodular maximization problem with cardinality constraint, knapsack constraint and matroid constraint respectively. With the help of threshold technique and binary search technique, they obtained a $ ( 1 - e^{-1} ) $ -threshold greedy algorithm.
Gottschalk et al. \cite{2015Gottschalk} cleverly used the bidirectional greed technique, that is, according to the predetermined order of indicators, respectively from 0 and the upper certainty to rise and fall, iterated until the two are equal to output the final solution, and gave a 1/3 polynomial-time algorithm for for the non-monotonic modular function maximization problem on integer lattices.
Nong et al. \cite{2020Nong} further improved the approximate ratio to 1/2. Based on the optimal budget allocation problem, Soma et al. \cite{2014Soma} established a monotone lattice submodular maximum on integer lattices with knapsack constraints
The approximation ratio is $( 1 - e^{-1} )$.
By using the improved threshold technique, Tan et al. \cite{2023Tan} proposes an online threshold streaming algorithm with 1/3-approximate ratios for the lattice submodular maximization problems with knapsack constraints in a flow model.

Although the submodular function maximization problem is widely used, there are still many problems that cannot be described well. The development of non-submodular function maximization problem is an important progress in the field of mathematical modeling and optimization, which is of great significance for both theoretical research and practical application. The research of non-submodular function maximization is also expanding to new application areas, such as social network influence maximization and resource allocation optimization. These applications are advancing the field and inspiring new research directions \cite{2012KM,2007S,2014SK}.
At first, scholars used submodular rate which quantifies the difference between the function and the submodular function to describe the non-submodular function.
The concept of submodular rate for a set function was initially proposed by
Das et al. \cite{2011Das}.
Bian et al. \cite{2017Bian} developed a greedy algorithm for non-submodular maximization problems with cardinality constraints on finite sets using submodular rates. By introducing the generalized curvature $c$ of the set function, they achieved an approximation ratio of
$ ( 1 - e^{ \alpha c } ) / c $ and proved its compactness.
Liu et al. \cite{2021Liu} studied the problem of maximizing the sum of a monotone weak submodular function and a supermodular function, providing both offline and streaming algorithms for it.
Kuhnle et al. \cite{2018Kuhnle} extended the notion of submodular rate to lattice submodular function and designed a threshold greedy algorithm on an integer lattice to solve this problem efficiently.
Zhang et al., \cite{2019Zhang} devised a greedy algorithm that incorporates curvature to improve the approximation ratio significantly.
Using threshold technique and lattice binary search technique, Tan \cite{2022Tan} proposes a threshold flow algorithm for the above problems.

Inspired by the above results, we design an efficient bicriteria streaming algorithms to construct a balance between gain and cost in a team formation problem with cardinality constraint on the integer lattice.
Let $E$ of size $n$ be the ground set,   $\boldsymbol{x}$ be an $n$-dimensional vector of $\mathbb{N}^{E} $, and $g$ be a  normalized submodular  function defined on $\mathbb{N}^{E} $ which satisfies nonnegativity and monotonicity, $c$ be a  linear  function defined on $\mathbb{N}^{E} $.
Then, the problem can be described briefly as below
\bea\label{main problem1}\nn
& \mathrm{maximize}&   g(\boldsymbol{x})-c(\boldsymbol{x}),\\
& \mathrm{s. t. } &  \boldsymbol{x}\leq \boldsymbol{b}, \\\nn
& &  \boldsymbol{x}(E)\leq k.
\eea
where the total budget is $k \in \mathbb{N}$, and $\boldsymbol{B}=\{\boldsymbol{x}\in \mathbb{N}^{E}: \boldsymbol{x}\leq\boldsymbol{b}\}$  is a box in
$\{ \mathbb{N}\bigcup \
\{+\infty\} \} ^{E}$.

For the problem, we first design Algorithm \ref{alg1} when the objective function is a monotone non-negative submodular function. In Algorithm 1, we call a binary search algorithm to decide whether to keep the currently arrived element and set a level for the remaining elements in memory. The effect of algorithm \ref{alg2} on the time complexity of the whole algorithm is so slight that we can ignore it.
Further, we discuss the case where the objective function is a non-submodular function, and we design a streaming algorithm \ref{alg3} with the help of lattice submodular rate, and analyze the performance of the algorithm.
Moreover,  the space complexity  and the  memory of the whole algorithm related to the lattice submodular rate, which is $ O ( k \log ( k / \alpha ) / \varepsilon)$.

We will divide the main results of this paper into three parts. In section \ref{sec2}, we will elaborate on the definitions and properties of integer lattices and submodules. In section \ref{sec3}, we design two combination algorithms to solve the above problems and analyze the performance of the two algorithms. Finally, we summarize the main results in the \ref{sec4} section.

\section{ Preliminaries}
\label{sec2}
In this section, we elucidate numerous symbols, conceptual definitions and axioms.

 $k$ is a natural number denoted as $k\in \mathbb{N^{+}}$ . We use $[k]$ to denote all positive integers between 1 and $k$. We establish a base set $E=\{e_1, e_2, \cdots, e_n\} $ and $\boldsymbol{x}$ as an $n$-dimensional vector in $\mathbb{N}^{E}$. The segment of coordinate $e_i\in E$ for $\boldsymbol{x}$ is represented as  $\boldsymbol{x}(e_i)$.
We employ $\boldsymbol{0}$ to depict the null vector and $\chi_{e_i}$ as the standard base vector. Thus, all components are at zero, except for the $i$-th component that stands at unity.
For any subset $ X\subseteq E$, its characteristic vector is $\chi_{X}$ and $\boldsymbol{x}(X)$  amounts to the summation of $\boldsymbol{x}(e_i)$ over all elements ei belonging to $\boldsymbol{x}$.
Within the set  $\boldsymbol{x}\in \mathbb{N}^{E}$, $supp^{+}(\boldsymbol{x})=\{e\in E|\boldsymbol{x}(e)>0 \}$ signifies the support set of $\boldsymbol{x}$.
We define a multi-set $\{\boldsymbol{x}\}$, where $e$ appears $\boldsymbol{x}(e)$ times, and use $| \{\boldsymbol{x}\}|:=\boldsymbol{x}(E)$ to express the total amount of $\boldsymbol{x}(E)$.
$\boldsymbol{x}\vee \boldsymbol{y}$ and $\boldsymbol{x}\wedge \boldsymbol{y}$  stand for the coordinate-wise maximum and minimum of vectors $\boldsymbol{x}$ and $\boldsymbol{y}$ respectively.
Multi-sets $\boldsymbol{x}$ and $\boldsymbol{y}$  have their difference denoted as: $\{\boldsymbol{x}\} \setminus \{\boldsymbol{y}\} := \{ (\boldsymbol{x} \setminus \boldsymbol{y}) \vee \boldsymbol{0} \}$

A function $ f: \mathbb{N} ^{E}\rightarrow \mathbb{R_+} $ is known as monotone non-decreasing if and only if  $ f(\boldsymbol{x}) \leq f(\boldsymbol{y})$ can be satisfied for any $\boldsymbol{x} \leq \boldsymbol{y}$.
A function f is deemed nonnegative and normalized if it invariably fulfills the conditions: $ f(\boldsymbol{x}) \geq 0 $ (for all $ \boldsymbol{x} \in \mathbb{N} ^{E} $ ) and $f(\boldsymbol{0})=0$.

\begin{defi}\label{def2.1}
A function $ f: \mathbb{N} ^{E} \rightarrow \mathbb{R}_{+} $ is termed as (lattice) submodular if $f$ fulfills this particular inequality for every $\boldsymbol{x}$ and $\boldsymbol{y}$ in $\mathbb{N} ^{E} $:
$$ f(\boldsymbol{y} + \chi_e) - f(\boldsymbol{y}) \le f(\boldsymbol{x} + \chi_e) - f(\boldsymbol{x}). $$
\end{defi}

\begin{defi}\label{def2.2}
Informally known as DR-submodular, a function $ f: \mathbb{N} ^{E}\rightarrow \mathbb{R}_{+} $ is thought of as diminishing return submodular if it qualifies the following inequality for any $ e \in E$, and any vector $\boldsymbol{x}$ and $\boldsymbol{y}$ in $\mathbb{N} ^{E} $ while $ \boldsymbol{x} \leq \boldsymbol{y} $:
$$ f(\boldsymbol{y} + \chi_e) - f(\boldsymbol{y}) \le f(\boldsymbol{x} + \chi_e) - f(\boldsymbol{x}). $$
\end{defi}

The value denoted by $ f( \boldsymbol{x}|\boldsymbol{y} ) $ presents the marginal increase of a vector $\boldsymbol{x}$ concerning $\boldsymbol{y}$ expressed as:
$$ f( \boldsymbol{x} | \boldsymbol{y} ) = f( \boldsymbol{x} + \boldsymbol{y} ) - f( \boldsymbol{y} ). $$.

Lastly, denote $\boldsymbol{x}^{*}$ as the ideal solution vector.

\section{ Bicriteria Streaming Algorithms}
\label{sec3}
For the case where the first function is submodular and non-submodular, we will design two combinatorial algorithms and analyze the performance of the two  respectively in this section.

\subsection{ Maximizing $ G - C $  function}
\label{subsec3.1}
First, we give a streaming algorithm for G-C on the integer lattice which is stated in Algorithm \ref{alg1}. In the process of running the algorithm, we not only need to decide whether the current element is preserved, but also count the level of the elements that are preserved. At this point, we need to call Algorithm 2 to perform the calculation.
The Algorithm \ref{alg2} details are as follows.

Denote$\boldsymbol{x^{*}}$ be the optimal solution, $\boldsymbol{x}$ is the final solution vector return by Algorithm \ref{alg1}. $p=\boldsymbol{x^{*}}(E)$, $q=\boldsymbol{x}(E)$. Assume that $\boldsymbol{x}_{i-1}$ is the set when the element $e_i$ arrives, $i= 1, 2, \cdots, q$, $\boldsymbol{x}_0=\boldsymbol{0}$, $\boldsymbol{x}_q=\boldsymbol{x}$. There are two case $\boldsymbol{x}(E)=k$(lem1) and $\boldsymbol{x}(E)<k$(lem2).

  \begin{algorithm}[h]
     \caption{Streaming algorithm for G-C on the integer lattice }\label{alg1}
     \begin{algorithmic}[1]
     \REQUIRE   Data stream $E$,  $ g \in \mathcal{\mathcal{G}}_b $, $ c \in \mathcal{\mathcal{C}}_b $, cardinality constraint $k\in \mathbb{N}_+$,  threshold value $\tau$,
                $t \geq 1$.
     \ENSURE a vector $\boldsymbol{x}\in \mathbb{N}^{E}$.

     \STATE {$ \boldsymbol{x}\leftarrow \boldsymbol{0}$;}
        \FOR {$e\in E$}
           \IF{ $ \boldsymbol{x}(E) < k$,  }
                \STATE {find the level $l$ with $ \textbf{BS Algorithm} \bigg( g,c, \boldsymbol{x}, \boldsymbol{b}, e, k, \tau \bigg); $}
             \IF{ $ \boldsymbol{x}+l\chi_{e_i}(E)\leq k$, }
                 \STATE {update $\boldsymbol{x}\leftarrow \boldsymbol{x}+l\chi_{e_i}; $}
             \ENDIF
           \ENDIF
         \ENDFOR
     \RETURN {$\boldsymbol{x}$}
     \end{algorithmic}
  \end{algorithm}

  \begin{algorithm}[h]
    \caption{BS Algorithm $(g, c, \boldsymbol{x}, \boldsymbol{b}, e, k, \tau)$ }
    \label{alg2}
    \begin{algorithmic}[1]
        \REQUIRE {Data stream $E$, $e\in E$,  submodular function $g\in \mathcal{\mathcal{G}}_b$, $c\in \mathcal{\mathcal{C}}_b$,   $\boldsymbol{x} \in \mathbb{N}^{E}$,  and $\tau \in \mathbb{R}_+$.}
        \ENSURE {$l$}

        \STATE {$ l_m\leftarrow 1 $;}
        \STATE {$l_n\leftarrow \min \bigg\{\boldsymbol{b}(e)-\boldsymbol{x}(e), k-\boldsymbol{x}(e)\bigg\}$;}
           \IF {$g(\chi_e|\boldsymbol{x})- t c(\chi_e) < \tau$}
               \STATE {return $0$.}
           \ENDIF
           \IF {$g(l_n\chi_e|\boldsymbol{x})- t c(l_n\chi_e)  \geq \tau$}
               \STATE {return $l_n$}
           \ENDIF
           \WHILE{ $l_m<l_n+1$,}
                  \STATE { $a=\lfloor \frac{l_n+l_m}{2}\rfloor$ }
               \IF{$g(l_n\chi_e|\boldsymbol{x})- t c(l_n\chi_e)  \geq \tau$ }
                   \STATE{$l_m=a,$}
               \ELSE{}
                   \STATE{$l_n=a,$}
               \ENDIF
           \ENDWHILE
           \RETURN {$l_m$}
      \end{algorithmic}
   \end{algorithm}

\begin{lem}\label{lem3.1}
    Let $ \boldsymbol{x} $ be the output of Algorithm \ref{alg1}.
    If $ \boldsymbol{x} (E) = k $, then the value of the function $ g ( \boldsymbol{x} ) - c ( \boldsymbol{x} ) $ is not less than $ k\tau $.
\end{lem}

\begin{proof}
   According to Algorithm \ref{alg2}, for any  element $ e_i \in \boldsymbol{x} $, we get that the following inequality is satisfied
   $$ g( l_i \chi_{ e_i } | \boldsymbol{x}_{i-1} ) - t c ( l_i \chi_{ e_i } ) \geq \tau, $$
   summing up all the inequalities above, we get
   we obtain $$ \sum \limits_{ e_i \in \{ \boldsymbol{x} \} }  g ( l_i \chi_{ e_i } | \boldsymbol{x}_{i-1} )
                   - t \sum \limits_{ e_i \in \{ \boldsymbol{x} \} }  c ( l_i \chi_{ e_i } )
               \geq \sum \limits_{ e_i \in \{ \boldsymbol{x} \} } \tau
               = \tau \boldsymbol{x} (E).$$
   With the aid of the submodularity of the function $ g $ and the linearity of  function $ c $, we can rewrite the above inequality as following
    $$ g ( \boldsymbol{x} ) - t c ( \boldsymbol{x} ) \geq \tau \boldsymbol{x} (E) = \tau k. $$
   Since $ t \geq 1 $,
   we get $$ g ( \boldsymbol{x} ) - c ( \boldsymbol{x} ) \geq g ( \boldsymbol{x} ) - t c ( \boldsymbol{x} ) \geq  \tau k. $$

   which completes the proof.
\end{proof}

\begin{lem}\label{lem3.2}
   Suppose $ \boldsymbol{x} $ is the output solution of Algorithm \ref{alg1} and satisfied $ \boldsymbol{x}(E) < k $, then the objective function value satisfies the following inequality
   $$ g ( \boldsymbol{x} ) - c ( \boldsymbol{x} ) \geq  \frac{ t - 1 } { t } g ( \boldsymbol{ x^{*} } )
      - ( t-1 ) c ( \boldsymbol{ x^{*} } )
      + k \tau ( \frac{1}{t} \beta - \frac{ t - 1 } { t } \alpha ), $$
   where $ 0 < \alpha \leq 1 $, $ 0 < \beta \leq 1 $, $ t \geq 1 $.
\end{lem}

\begin{proof}
     First, we denote
     $ \{ \bar{ \boldsymbol{x} } \} = \{ \boldsymbol{ x^{*} } \} \setminus \{ \boldsymbol{x} \} $,
     $ s = \bar{ \boldsymbol{x} } (E)$.
     let $  \boldsymbol{x}_{i-1} $ be the corresponding output solution for each element $ e_i $ $ ( i \in \{ 1, \cdots, s \} ) $ that $ \chi_{ e_i } $ has arrived but not been added to $ \boldsymbol{x} $.

     Next, we will prove the lower bound of function $ g (  \boldsymbol{x}  ) $ and the upper bound of $ c (  \boldsymbol{x}  ) $ respectively.
   Now, We proof the the lower bound of $ g (  \boldsymbol{x}  ) $. According to Algorithm \ref{alg2}, for each element $ e_i \in \{ \bar{ \boldsymbol{x} } \} $,
   we have $$ g ( \chi_{ e_i } \mid \boldsymbol{x}_ { i-1 } ) - t c ( \chi_{ e_i } ) < \tau.$$
   Summing up all the above inequalities of  the elements in
   $ \bar{ \boldsymbol{x} } $, we obtain
     \bea\label{eq3.1.1}
          \sum \limits_{ e_i \in \{ \boldsymbol{x} \} }  g (  \chi_{ e_i } | \boldsymbol{x}_{i-1} )
          - t \sum \limits_{ e_i \in \{ \boldsymbol{x} \} }  c ( \chi_{ e_i } )
        < \sum \limits_{ e_i \in \{ \boldsymbol{x} \} } \tau
        = \tau | \{ \bar{ \boldsymbol{x} } \} |
     \eea
   Using the submodularity of the function  $ g $, we have
      \bea
         \nn     & &      \sum \limits_{ e_i \in \{ \bar{ \boldsymbol{x} } \} } g (  \chi_{ e_i } | \boldsymbol{x}_{i-1} )
         \\ \nn & \geq &  \sum \limits_{ e_i \in \{ \bar{ \boldsymbol{x} } \} } g (  \chi_{ e_i } | \boldsymbol{x} )
         \\ \nn & \geq &  g ( \bar{ \boldsymbol{x} } | \boldsymbol{x} )
         \\ \nn &  =   &  g ( \bar{ \boldsymbol{x} } + \boldsymbol{x} ) - g ( \boldsymbol{x} )
         \\ \nn & \geq &  g ( \boldsymbol{ x^{*} } + \boldsymbol{x} ) - g ( \boldsymbol{x} ).
       \eea
   Since function $ c $ is non-negative linear, we get
      $$ \sum \limits_{ e_i \in \{ \bar{ \boldsymbol{x} } \} } c (  \chi_{ e_i } ) = c (  \bar{ \boldsymbol{x} } ) \leq c ( \boldsymbol{ x^{*} } ). $$
   By combining the above two inequalities with the monotonicity of the function $ g $, we can derive the following inequality
      $$ g ( \boldsymbol{ x^{*} } ) - g ( \boldsymbol{x} ) - t c ( \boldsymbol{ x^{*} } ) \leq \tau | \bar{ \boldsymbol{x} } (E) |, $$
  that is,
      \bea  \label{3.1.3}
       g ( \boldsymbol{x} )
       \geq  g ( \boldsymbol{ x^{*} } ) - t c ( \boldsymbol{ x^{*} } ) - \tau | \bar{
                  \boldsymbol{x} } (E) |.
      \eea

   Moreover, We can get the upper bound of a function $ c ( \boldsymbol{x} ) $ that is related to the function $ g ( \boldsymbol{x} ) $.
   From the running process of the algorithm, we can see that all the elements selected into the output solution  $ \boldsymbol{x} $ satisfy the following inequality
      $$  g ( l_i \chi_{ e_i } | \boldsymbol{x}_{i-1} ) - t c ( l_i \chi_{ e_i } ) \geq \tau. $$
   Summing up above all the inequality which the element $ e_i \in \{ \boldsymbol{x} \} $, we get
      $$ \sum \limits_{ e_i \in \{ \boldsymbol{x} \} } g ( l_i \chi_{ e_i } | \boldsymbol{x}_{i-1} )
            - t \sum \limits_{ e_i \in \{ \boldsymbol{x} \} } c ( l_i \chi_{ e_i } )
        \geq \sum \limits_{ e_i \in \{ \boldsymbol{x} \} } \tau
         =  \tau \boldsymbol{x} (E). $$
   Under the assumption that the function $ g ( \boldsymbol{x} ) $is a monotone submodular function and that the function $ c $ is a linear function, we can estimate the term of the function $ c ( \boldsymbol{x} ) $, that is,
      \bea  \label{3.1.4}
       c ( \boldsymbol{x} ) \leq \frac{1}{t} g ( \boldsymbol{x} ) - \frac{\tau}{t} \boldsymbol{x}(E).
       \eea
   Then, by combining inequality (\ref{3.1.3}) and inequality (\ref{3.1.4}), we get the lower bound of the value of objective function $ g ( \boldsymbol{x} ) - c ( \boldsymbol{x} ) $ as follows
       \bea \label{3.1.5}\nn
         & & g ( \boldsymbol{x} ) - c ( \boldsymbol{x} )
        \\ \nn & \geq &  g ( \boldsymbol{x} ) -  \frac{1}{t} g ( \boldsymbol{x} ) + \frac{\tau}{t} \boldsymbol{x}(E)
        \\  & \geq & \frac{t-1}{t} g ( \boldsymbol{x^{*}} ) - (t-1) c ( \boldsymbol{x^{*}} )
                      - \frac{t-1}{t} \bar{ \boldsymbol{x} }(E) \tau + \frac{\tau}{t} \boldsymbol{x} (E)
       \eea
   For the rest of the proof, we do this by introducing two parameters $ \mu ( 0 \leq \mu \leq 1 ) $ and $ \nu ( 0 \leq \mu \leq 1 )$ that depend on the final output set or the current set.  We assume without loss of generality that $ \bar{ \boldsymbol{x} }(E) = \mu k $, $ \boldsymbol{x} (E) = \nu_k $.  During algorithmic analysis, we only use $ \mu $ and $ \nu $ to refer to multiple factors, so they are need to consistent with other parameters in the current analysis.  Then we can rewrite the inequality (\ref{3.1.5})
       \bea \nn
         & & g ( \boldsymbol{x} ) - c ( \boldsymbol{x} )
        \\ \nn & \geq & \frac{t-1}{t} g ( \boldsymbol{x^{*}} ) - (t-1) c ( \boldsymbol{x^{*}} )
                      - \frac{t-1}{t}\mu k \tau + \frac{\tau}{t} \boldsymbol{x} (E)
        \\ \nn & = &   \frac{t-1}{t} g ( \boldsymbol{x^{*}} ) - (t-1) c ( \boldsymbol{x^{*}} )
                     +  k \tau ( \frac{1}{t} \nu - \frac{t-1}{t} \mu )
       \eea
   So, we can come to the conclusion that the output set $ \boldsymbol{x} $ in Lemma \ref{lem3.2} satisfies
       $$     g ( \boldsymbol{x} ) - c ( \boldsymbol{x} )
        \geq  \frac{t-1}{t} g ( \boldsymbol{x^{*}} ) - (t-1) c ( \boldsymbol{x^{*}} )
                     +  k \tau ( \frac{1}{t} \nu - \frac{t-1}{t} \mu ), $$
   for $ 0 \leq \mu \leq 1, 0 \leq \nu \leq 1, t \geq 1.$

      which completes the proof.
\end{proof}

\begin{thm}\label{thm3.1}
   The bicriteria ratio for  Algorithm \ref{alg1} is
   $ ( \frac{ t-1 }{ t + \mu ( t-1 ) - \nu }, \frac{ t-1 }{  t + \mu ( t-1 ) - \nu } \cdot t ) $,
   where $ t = \frac{ 2 + \mu + \sqrt{ \delta} } { 2 } $, $ \delta= \mu^{2} + 4 - 4 \nu $, $ 0 \leq \mu \leq 1, 0 \leq \nu \leq 1$.
\end{thm}

\begin{proof}
   Lemma \ref{lem3.1} and Lemma \ref{lem3.2} show that for any $0 \leq \mu \leq 1, 0 \leq \nu \leq 1, t \geq 1,$   the objective value satisfies
     $$ g ( \boldsymbol{x} ) - c ( \boldsymbol{x} )  \geq k \tau ,  \boldsymbol{x}(E) = k  $$
   and
    $$     g ( \boldsymbol{x} ) - c ( \boldsymbol{x} )
        \geq  \frac{t-1}{t} g ( \boldsymbol{x^{*}} ) - (t-1) c ( \boldsymbol{x^{*}} )
                     +  k \tau ( \frac{1}{t} \nu - \frac{t-1}{t} \mu ),
                     \boldsymbol{x}(E) \geq k$$
  for $ 0 \leq \mu \leq 1, 0 \leq \nu \leq 1, t \geq 1.$
  Therefore, the value of the final objective function is minimum of
   $  k \tau $ and $ \frac{t-1}{t} g ( \boldsymbol{x^{*}} ) - (t-1) c ( \boldsymbol{x^{*}} ) +  k \tau ( \frac{1}{t} \nu - \frac{t-1}{t} \mu )\} $.
  Let
  $$ k \tau = \frac{t-1}{t} g ( \boldsymbol{x^{*}} ) - (t-1) c ( \boldsymbol{x^{*}} ) +  k \tau ( \frac{1}{t} \nu - \frac{t-1}{t} \mu ),$$
  we obtain $$ k \tau = \frac{ t-1 }{ t + \mu ( t - 1 ) - \nu } g ( \boldsymbol{x^{*}} )
                       - \frac{ t-1 }{ t + \mu ( t - 1 ) - \nu } \cdot t \cdot c ( \boldsymbol{x^{*}} ).$$
  We calculate the value of $ t $ such that the coefficient $ c ( \boldsymbol{x^{*}} ) $ becomes 1, that is,
     $$ \frac{ t-1 }{ t + \mu ( t - 1 ) - \nu } \cdot t = 1 $$
  We can get two values for $ t $ from the above equality
    $$ t_1 = \frac{ 2 + \mu + \sqrt{ \delta } } { 2 } > 1 + \mu, $$
  and
    $$ t_2 = \frac{ 2 + \mu - \sqrt{ \delta } } { 2 } < 1 + \mu, $$
  where $ \delta = ( 2 + \mu )^{2} - 4 ( \mu + \nu ) = 4 + \mu^{2} - 4 \nu $ and $ \delta > 0 $, $ \sqrt{ \delta } > \mu $ under the assumption $ 0 \leq \mu \leq 1 $ and $ 0 \leq \nu \leq 1  $.

  Depending on the range of $ t $, we take $ t = t_1 $. Then
    $$ k \tau = \frac{ t-1 }{ t + \mu ( t - 1 ) - \nu } g ( \boldsymbol{x^{*}} )
                       - \frac{ t-1 }{ t + \mu ( t - 1 ) - \nu } \cdot t \cdot c ( \boldsymbol{x^{*}} ).$$
  and the threshold value
    $$ \tau = \frac{1}{k} \frac{ t-1 }{ t + \mu ( t - 1 ) - \nu } g ( \boldsymbol{x^{*}} )
                       -  \frac{1}{k} \frac{ t-1 }{ t + \mu ( t - 1 ) - \nu } \cdot t \cdot c ( \boldsymbol{x^{*}} ).$$
  where $$ t = \frac{ 2 + \mu + \sqrt{ \delta } } { 2 } > 1 + \mu,$$
        $$ \delta =  4 + \mu^{2} - 4 \nu, $$
        $$ 0 \leq \mu \leq 1, $$
        $$ 0 \leq \nu \leq 1.$$
  Therefore, we can get the conclusion of the Theorem \ref{thm3.1} that the bicriteria ratio output by Algorithm \ref{alg1} is
   $ ( \frac{ t-1 }{ t + \mu ( t-1 ) - \nu }, \frac{ t-1 }{  t + \mu ( t-1 ) - \nu } \cdot t )  $.
\end{proof}

  From the  previous discussion of Theorem \ref{thm3.1}, we have
   $$ \tau = \frac{1}{k} \frac{ t-1 }{ t + \mu ( t - 1 ) - \nu } g ( \boldsymbol{x^{*}} )
                       -  \frac{1}{k} \frac{ t-1 }{ t + \mu ( t - 1 ) - \nu } \cdot t \cdot c ( \boldsymbol{x^{*}} ).$$
  So, we need to caculate the optimal value of function
   $$ f( \boldsymbol{x} ) = \frac{ t-1 }{ t + \mu ( t - 1 ) - \nu } g ( \boldsymbol{x} )
                       - \frac{ t-1 }{ t + \mu ( t - 1 ) - \nu } \cdot t \cdot c ( \boldsymbol{x} ). $$
  In order to estimate the value $ f( \boldsymbol{x}^{*} ) $, we take the means of
  the maximum singleton value $ m = \max\limits_{ e \in E } f( \chi_e ) $ instead which
  inspired by the ideas of Badanidiyuru et al.\cite{2014Badanidiyuru} and Ene \cite{2020Ene}. Also, form their discussion above the OPT, we can ensure that $ m \le OPT \le k m $. It takes $ O ( \log k / \varepsilon )$ values falling in $ [ m, km ]$ to get a $ ( 1 + \varepsilon ) $- approximation to OPT. Unfortunately, with this approach we need to read the whole stream data two passes. In order to reduce the number of data reads to one, we relax  $m$ to the maximum singleton value of all arrived elements.
  When different thresholds are selected, we can parallel run the original algorithm with a run count of $ O ( \log k / \varepsilon )$.
  Therefore, the total memory is $ O ( k \log k / \varepsilon )$.

  In addition, we discuss the query number for per element. During the run of Algorithm \ref{alg1}, we need call the Algorithm \ref{alg2} to determine the elements retained in the storage and the level of each preserved element, that is, any $l$ should satisfy both $\frac{ f( l \chi_e | \boldsymbol{x} )}{l} \geq \tau $ and $f( \chi_e | \boldsymbol{x} + l \chi_e ) < \tau $. In order to obtain a valid level $l$, we need to query function $f$ with $O(\log K)$  times in Algorithm \ref{alg2}.  So, we can obtain that the query number for per element is $ O(\log K)$.

\subsection{Streaming Algorithm for  $\alpha G - C $ }\label{subsec3.2}
In this section, we shall discuss the case where the first function $ g $ in the objective function does not satisfy submodularity. We assume that the first function $ g $ in the objective satisfies $ \alpha $-weakly submodular, non-negative, and monotoneand.  The second function $ c $ in the objective satisfies non-negative linearity. The problem is described as follows:
     \bea\label{Pro2}\nn
        & \mathrm{maximize}&   g(\boldsymbol{x})-c(\boldsymbol{x}),\\
        & \mathrm{s. t. } &  \boldsymbol{x}\leq \boldsymbol{b}, \\ \nn
        & &  \boldsymbol{x}(E)\leq k,
     \eea
  We first  define $ \alpha $-weakly submodular function in the integer lattice.

  \begin{algorithm}[h]
     \caption{Streaming algorithm for $\alpha G-C $ on the integer lattice }\label{alg3}
     \begin{algorithmic}[1]
     \REQUIRE   Data stream $ E $,  non-negative normalized monotone $ \alpha $-weakly submodular function $ g $ ,  non-negative linear function $ c $, cardinality constraint $ k \in \mathbb{N}_+ $,  threshold value $ \tau_\alpha $, $ 0 < \alpha \geq 1$.
     \ENSURE a vector $\boldsymbol{x}\in \mathbb{N}^{E}$.

     \STATE {$ \boldsymbol{x}\leftarrow \boldsymbol{0}$;}
        \FOR {$e\in E$}
           \IF{ $ \boldsymbol{x}(E) < k$,  }
                \STATE {find the level $l$ with $ \textbf{BS Algorithm} \bigg( g,c, \boldsymbol{x}, \boldsymbol{b}, e, k, \tau_\alpha \bigg); $}
             \IF{ $ \boldsymbol{x}+l\chi_{e_i}(E)\leq k$, }
                 \STATE {update $\boldsymbol{x}\leftarrow \boldsymbol{x}+l\chi_{e_i}; $}
             \ENDIF
           \ENDIF
         \ENDFOR
     \RETURN {$\boldsymbol{x}$}
     \end{algorithmic}
  \end{algorithm}

  \begin{algorithm}[h]
    \caption{BS Algorithm $(g, c, \boldsymbol{x}, \boldsymbol{b}, e, k, \tau_\alpha)$ }
    \label{alg4}
    \begin{algorithmic}[1]
        \REQUIRE {Data stream $E$, $e\in E$,  submodular function $g\in \mathcal{\mathcal{G}}_b$, $c\in \mathcal{\mathcal{C}}_b$,   $\boldsymbol{x} \in \mathbb{N}^{E}$,  and $\tau_\alpha \in \mathbb{R}_+$.}
        \ENSURE {l.}

        \STATE {$ l_m\leftarrow 1 $;}
        \STATE {$l_n\leftarrow \min \bigg\{\boldsymbol{b}(e)-\boldsymbol{x}(e), k-\boldsymbol{x}(e)\bigg\}$;}
           \IF {$g(\chi_e|\boldsymbol{x})- ( 1 + \alpha ) c(\chi_e) < \tau_\alpha$}
               \STATE {return $0$.}
           \ENDIF
           \IF {$g(l_n\chi_e|\boldsymbol{x})- ( 1 + \alpha ) c(l_n\chi_e)  \geq \tau_\alpha$}
               \STATE {return $l_n$}
           \ENDIF
           \WHILE{ $l_m<l_n+1$,}
                  \STATE { $a=\lfloor \frac{l_n+l_m}{2}\rfloor$ }
               \IF{$g(l_n\chi_e|\boldsymbol{x})- ( 1 + \alpha ) c(l_n\chi_e)  \geq \tau_\alpha$ }
                   \STATE{$l_m=a,$}
               \ELSE{}
                   \STATE{$l_n=a,$}
               \ENDIF
           \ENDWHILE
           \RETURN {$l_m$}
      \end{algorithmic}
   \end{algorithm}

 \begin{defi} \cite {2018Kuhnle} \label{def2.3}
    For any $ \boldsymbol{s}, \boldsymbol{t} \in D _{ \boldsymbol{c} } $ with
    $ \boldsymbol{s} \leq \boldsymbol{t} $ and $ e\in G $, the $ \alpha $-weakly submodular function  $g$ in $ G_{ \boldsymbol{c} } $
    is the maximum scalar $ \alpha_g  $  such that
    $$ \alpha_g  g( \chi_{e} | \boldsymbol{t} ) \leq g( \chi_{e} | \boldsymbol{s} )$$
    where $ \boldsymbol{t} + \chi_e \in E_{ \boldsymbol{c} } $.
 \end{defi}

  In Algorithm \ref{alg3}, the first step involves setting a specific threshold value for the element present in the data stream. This becomes our preliminary standard when determining whether these data points should be incorporated into the current solution set. However, this decision-making process is not solely reliant on the set threshold. We also refer to Algorithm \ref{alg4}, which helps us finalize the exact level (denoted as $l$ here) of the current element that should be included in the solution. To better understand this, consider reaching an element in the stream which the value of $g(\chi_e|\boldsymbol{x})- ( 1 + \alpha ) c(\chi_e) $ does not exceed the previously established threshold, under such circumstances, we follow through with the decision to discard that element, choosing not to maintain it within our current solution. The rationale behind this decision arises from the fact that this element does not meet the defined threshold condition, indicating that its contribution to resolving the problem is likely insignificant or counterproductive. Conversely, if the same evaluation results in a value exceeding the threshold, the element is deemed essential and retained within the solution. When it comes to determining the number of the element that can be kept, we rely on inequality
   $$g(l_n\chi_e|\boldsymbol{x})- ( 1 + \alpha ) c(l_n\chi_e)  \geq \tau_\alpha.$$
  The intricacies of these mechanisms and steps are better explained and extensively covered within the respective descriptions of Algorithm \ref{alg3} and Algorithm \ref{alg4}.

  The analytical approach is similar to the proof methodology detailed in Sect.\ref{subsec3.1}. To make explanations clearer and easier to understand, we have based our discussions on several assumptions.  First, we define $\boldsymbol{x^{*}}$ to be an optimal solution, with $\boldsymbol{x}$ representing the final solution vector generated by Algorithm \ref{alg3}. Subsequently, from the vectors mentioned earlier, we identify two key values: $p$ and $q$. Here, $p$ equals the value of $\boldsymbol{x^{*}}(E)$, while $q$ corresponds to the value of $\boldsymbol{x}(E)$.
  Following this, we suppose $\boldsymbol{x}_{i-1}$ represents the set when element $e_i$ arrives, for each instance of $i$ ranging from 1 to $q$, $\boldsymbol{x}_0=\boldsymbol{0}$, $\boldsymbol{x}_q=\boldsymbol{x}$.
  Within the context of this discussion, we mainly examine two special scenarios:
  When $\boldsymbol{x}(E)$ equals $k$, our discussion primarily revolves around Lemma \ref{lem3.3}.
  When $\boldsymbol{x}(E)$ is less than $k$. our discussion corresponding to this case are included in Lemma \ref{lem3.4}.

  \begin{lem}\label{lem3.3}
     Assume $\boldsymbol{x}$  is the solution output by Algorithm \ref{alg3}, when $\boldsymbol{x}(E)=k$, then the objective function value satisfies  $ g ( \boldsymbol{x} ) - c ( \boldsymbol{x} ) \geq k \tau_\alpha $.
  \end{lem}

  \begin{proof}
   For each element selected into the output solution, according to the selection rule in Algorithm \ref{alg3}, the following inequality must be satisfied
    $ g( l_i \chi_{ e_i } | \boldsymbol{x}_{i-1} ) - ( 1 + \alpha ) c ( l_i \chi_{ e_i } ) \geq \tau_\alpha $.
    For all elements selected into the output solution $ { \boldsymbol{x} } $,  by adding the inequalities they satisfy, we can obtain the following comprehensive inequality
  $$ \sum \limits_{ e_i \in \{ \boldsymbol{x} \} }  g ( l_i \chi_{ e_i } | \boldsymbol{x}_{i-1} )
                   - ( 1 + \alpha ) \sum \limits_{ e_i \in \{ \boldsymbol{x} \} }  c ( l_i \chi_{ e_i } )
               \geq \sum \limits_{ e_i \in \{ \boldsymbol{x} \} } \tau_\alpha
               = \tau_\alpha \boldsymbol{x} (E).$$
   Since function $ g $ is submodular and function $ c $ is linear, rearranging the inequality,
   By combining our knowledge about the submodular property of function $ g $  and the linear behavior of function $ g $ , we can deduce the following inequality
   $$ g ( \boldsymbol{x} ) - ( 1 + \alpha ) c ( \boldsymbol{x} ) \geq \tau_\alpha \boldsymbol{x} (E) = \tau_\alpha k. $$
   Since $ t \geq 1 $,
   we can draw the following conclusions
    $$ g ( \boldsymbol{x} ) - c ( \boldsymbol{x} ) \geq g ( \boldsymbol{x} ) - ( 1 + \alpha ) c ( \boldsymbol{x} ) \geq  k \tau_\alpha. $$

   The proof of the conclusion of Lemma \ref{lem3.3} is completed.
\end{proof}

\begin{lem}\label{lem3.4}
    Suppose $ \boldsymbol{x} $ is the output solution of Algorithm \ref{alg3} and satisfied $ \boldsymbol{x}(E) < k $, then the objective function value satisfies the following inequality
    $$ g ( \boldsymbol{x} ) - c ( \boldsymbol{x} ) \geq  \frac{ \alpha } { 1 + \alpha } g ( \boldsymbol{ x^{*} } )
      -  c ( \boldsymbol{ x^{*} } )
      +  \frac{ \nu - \mu } { 1 + \alpha }  k \tau_\alpha , $$
   where $ 0 < \mu \leq 1 $, $ 0 < \nu \leq 1 $, $ t \geq 1 $.
\end{lem}

\begin{proof}
   Suppose
     $ \{ \bar{ \boldsymbol{x} } \} = \{ \boldsymbol{ x^{*} } \} \setminus \{ \boldsymbol{x} \} $,
     $ l = \bar{ \boldsymbol{x} } (E)$.
   For each $ i \in \{ 1, \cdots, l \} $, let $  \boldsymbol{x}_{i-1} $ be the corresponding output solution when the current element $ e_i $ is arrived but $ \chi_{ e_i } $ not added to the output solution. Next, we will separately prove the lower bound of $ g ( \bar { \boldsymbol{x} } ) $ and the upper bound of $ c ( \boldsymbol{x} ) $.

   According to the rule of selecting elements in Algorithm \ref{alg3}, for each element $ e_i \in \{ \bar{ \boldsymbol{x} } \} $,
   we have
   $$ g ( \chi_{ e_i } \mid \boldsymbol{x}_ { i-1 } ) - ( 1 + \alpha ) c ( \chi_{ e_i } ) < \tau_\alpha .$$
   For all elements selected into the output solution $ e_i \in \{ \bar{ \boldsymbol{x} } \} $, by adding the inequalities they satisfy, we can obtain the following comprehensive inequality
     \bea\label{eq3.1.1}
          \sum \limits_{ e_i \in \{ \boldsymbol{x} \} }  g (  \chi_{ e_i } | \boldsymbol{x}_{i-1} )
          -  ( 1 + \alpha ) \sum \limits_{ e_i \in \{ \boldsymbol{x} \} }  c ( \chi_{ e_i } )
        < \sum \limits_{ e_i \in \{ \boldsymbol{x} \} } \tau_\alpha
        = \tau_\alpha | \{ \bar{ \boldsymbol{x} } \} |
     \eea
   Further, based on the $ \alpha $ - weakly submodular of function $ g $, we can derive
      $$    \sum \limits_{ e_i \in \{ \bar{ \boldsymbol{x} } \} } g (  \chi_{ e_i } | \boldsymbol{x}_{i-1} )
        \geq  \alpha g ( \bar{ \boldsymbol{x} } |  \boldsymbol{x}  )
        \geq \alpha g ( \bar{ \boldsymbol{x} } +  \boldsymbol{x} ) - g ( \boldsymbol{x} ). $$
   Consider the monotonicity of function $ g $, we have
   $$     g ( \boldsymbol{ x^{*} } ) - g ( \boldsymbol{x} )
      \leq g ( \boldsymbol{ x^{*} } + \boldsymbol{x} ) - g ( \boldsymbol{x} )
         = g ( \bar{ \boldsymbol{x} } +  \boldsymbol{x} ) - g ( \boldsymbol{x} )
   $$
   Since function $ c $ is non-negative linear, we get
      $$ \sum \limits_{ e_i \in \{ \bar{ \boldsymbol{x} } \} } c (  \chi_{ e_i } ) = c (  \bar{ \boldsymbol{x} } ) \leq c ( \boldsymbol{ x^{*} } ). $$
   Considering the inequalities satisfied by $ \sum \limits_{ e_i \in \{ \boldsymbol{x} \} }  g (  \chi_{ e_i } | \boldsymbol{x}_{i-1} ) $ and
      $ \sum \limits_{ e_i \in \{ \boldsymbol{x} \} }  c ( \chi_{ e_i } ) $ together, we can rewrite inequality $ ( \ref {eq3.1.1} ) $ as follows
      $$ \alpha ( g ( \boldsymbol{ x^{*} } ) - g ( \boldsymbol{x} ) ) - ( 1 + \alpha ) c ( \boldsymbol{ x^{*} } ) \leq \tau_\alpha | \bar{ \boldsymbol{x} } (E) |, $$
   which can be formulated as
      $$ g ( \boldsymbol{x} ) \geq g ( \boldsymbol{ x^{*} } ) - \frac{ 1 + \alpha }{ \alpha } c ( \boldsymbol{ x^{*} } )
      - \frac{ \tau_\alpha | \bar{ \boldsymbol{x} } (E) | } { \alpha}. $$
   Now, we are in the position to proof the upper bound of $ c ( \boldsymbol{x} ) $.

   According to Algorithm \ref{alg3}, all the elements in the output solution $ \boldsymbol{x} $ satisfy the following inequalities
    $$  g ( l_i \chi_{ e_i } | \boldsymbol{x}_{i-1} ) - ( 1 + \alpha ) c ( l_i \chi_{ e_i } ) \geq \tau_\alpha. $$
    Summing up all the above inequalities, we get
      $$ \sum \limits_{ e_i \in \{ \boldsymbol{x} \} } g ( l_i \chi_{ e_i } | \boldsymbol{x}_{i-1} )
            - ( 1 + \alpha ) \sum \limits_{ e_i \in \{ \boldsymbol{x} \} } c ( l_i \chi_{ e_i } )
        \geq \sum \limits_{ e_i \in \{ \boldsymbol{x} \} } \tau_\alpha
         =  \tau_\alpha \boldsymbol{x} (E). $$
   Due to the monotonicity and submodularity of function $ g $ and the linearity of function $ c $, we can transform the inequality as follows
      $$ c ( \boldsymbol{x} ) \leq \frac{ 1 }{ 1 + \alpha } g ( \boldsymbol{x} ) - \frac{ \tau_\alpha }{ 1 + \alpha } \boldsymbol{x}(E). $$
   So, we can get the bound of the value of objective function
       \bea \nn
         & & g ( \boldsymbol{x} ) - c ( \boldsymbol{x} )
        \\ \nn & \geq &  g ( \boldsymbol{x} ) -  \frac{1}{ 1 + \alpha } g ( \boldsymbol{x} ) + \frac{ \tau_\alpha }{ 1 + \alpha } \boldsymbol{x}(E)
        \\ \nn & = & \frac{\alpha}{ 1 + \alpha } g ( \boldsymbol{x} ) + \frac{ \tau_\alpha }{ 1 + \alpha } \boldsymbol{x}(E)
        \\ \nn & \geq & \frac{ \alpha }{ 1 + \alpha } ( g ( \boldsymbol{x^{*}} ) - \frac{ \alpha }{ 1 + \alpha } c ( \boldsymbol{x^{*}} )
                      - \frac{\tau_\alpha}{\alpha} \bar{ \boldsymbol{x} }(E) )+ \frac{ \tau_\alpha }{ 1 + \alpha } \boldsymbol{x}(E)
        \\ \nn & = & \frac{ \alpha }{ 1 + \alpha }  g ( \boldsymbol{x^{*}} ) - c ( \boldsymbol{x^{*}} )
                 + \frac{ \boldsymbol{x}(E) - \bar{ \boldsymbol{x} }(E) }{ 1 + \alpha } \boldsymbol{x}(E)
       \eea
   Similar to Lemma 2, we also define parameters $ \bar{ \boldsymbol{x} }(E) = \mu k $ and $ \boldsymbol{x} (E) = \nu_k $,
   by introducing parameters $ \mu $ and $ \nu $ with the range of
   values for $ \mu $ and $ \nu : 0 \leq \mu \leq 1, 0 \leq \nu \leq 1 $. Then we have
       \bea \nn
         & & g ( \boldsymbol{x} ) - c ( \boldsymbol{x} )
        \\ \nn & \geq & \frac{ \alpha }{ 1 + \alpha } g ( \boldsymbol{x^{*}} ) - c ( \boldsymbol{x^{*}} )
                      - \frac{ k \nu - k \mu }{ 1 + \alpha } \tau_\alpha
        \\ \nn & = &   \frac{ \alpha }{ 1 + \alpha } g ( \boldsymbol{x^{*}} ) - c ( \boldsymbol{x^{*}} ) +  \frac{  \nu -  \mu }{ 1 + \alpha } k \tau_\alpha
       \eea
   That is, the output solution $ \boldsymbol{x} $ in Lemma \ref{lem3.4} satisfies
       $$     g ( \boldsymbol{x} ) - c ( \boldsymbol{x} )
        \geq  \frac{ \alpha }{ 1 + \alpha } g ( \boldsymbol{x^{*}} ) -  c ( \boldsymbol{x^{*}} )
                     + \frac{  \nu -  \mu }{ 1 + \alpha } k \tau_\alpha, $$
   for $ 0 \leq \mu \leq 1, 0 \leq \nu \leq 1, t \geq 1.$

   The proof of the conclusion of Lemma \ref{lem3.4} is completed.
\end{proof}

\begin{thm}\label{thm3.2}
   The bicriteria ratio for  Algorithm \ref{alg3} is
   $ ( \frac{ \alpha }{ 1 + \alpha + \mu - \nu }, \frac{ 1 + \alpha }{  1 + \alpha + \mu - \nu }  ) $,
   where $ 0 < \alpha \leq 1 $,  $ 0 \leq \mu \leq 1, 0 \leq \nu \leq 1$.
\end{thm}

\begin{proof}
  Based on the cardinality of $ \boldsymbol{x} $, we start to prove the conclusion.
  Lemma \ref{lem3.3} and Lemma \ref{lem3.4} show that for any $0 \leq \mu \leq 1, 0 \leq \nu \leq 1, t \geq 1,$   the objective value satisfies
     $$ g ( \boldsymbol{x} ) - c ( \boldsymbol{x} )  \geq k \tau_\alpha ,\ \ \ \  \boldsymbol{x}(E) \geq k,$$
   and
    $$       g ( \boldsymbol{x} ) - c ( \boldsymbol{x} )
        \geq  \frac{ \alpha }{ 1 + \alpha } g ( \boldsymbol{x^{*}} )
         -   c ( \boldsymbol{x^{*}} ) +  k \tau_\alpha ( \frac{ \nu - \mu }{ 1 + \alpha } ),\ \ \ \ \ \boldsymbol{x}(E) \geq k. $$
  Therefore, the value of the final objective function is minimum of
   $k \tau_\alpha $ and $ \frac{ \alpha }{ 1 + \alpha } g ( \boldsymbol{x^{*}} ) -  c ( \boldsymbol{x^{*}} ) +  k \tau_\alpha ( \frac{ \nu - \mu }{ 1 + \alpha } ) \} $.
  Denote $ k \tau_\alpha = \frac{ \alpha }{ 1 + \alpha } g ( \boldsymbol{x^{*}} ) -  c ( \boldsymbol{x^{*}} ) +  k \tau_\alpha \frac{ \nu - \mu }{ 1 + \alpha } $,
  we have $$ k \tau_\alpha = \frac{ \alpha }{ 1 + \alpha + \mu  - \nu } g ( \boldsymbol{x^{*}} )
                       - \frac{ 1 + \alpha }{ 1 + \alpha + \mu  - \nu }  c ( \boldsymbol{x^{*}} ).$$
  So,  the threshold value
    $$ \tau_\alpha = \frac{1}{k} \frac{ \alpha }{ 1 + \alpha + \mu  - \nu } g ( \boldsymbol{x^{*}} )
                       -  \frac{1}{k} \frac{ 1 + \alpha }{ 1 + \alpha + \mu  - \nu }  c ( \boldsymbol{x^{*}} ).$$
  Therefor, the solution $ \boldsymbol{x} $ output by Algorithm \ref{alg3} satisfies
   $$ g ( \boldsymbol{x} ) - c ( \boldsymbol{x} )
        \geq  k \tau_\alpha
        \geq  \frac{ \alpha }{ 1 + \alpha + \mu  - \nu } g ( \boldsymbol{x^{*}} ) -  \frac{ 1 + \alpha }{ 1 + \alpha + \mu  - \nu } c ( \boldsymbol{x^{*}} ), $$
  so, we can obtain the bicriteria ratio  is
   $ ( \frac{ \alpha }{ 1 + \alpha + \mu  - \nu }, \frac{ 1 + \alpha }{  1 + \alpha + \mu  - \nu }  )  $,
   where $ 0 < \alpha \leq 1 $,  $ 0 \leq \mu \leq 1, 0 \leq \nu \leq 1$.
   \end{proof}

   We denote $$ h_\alpha ( \boldsymbol{x^{*}} ) \triangleq \frac{ \alpha }{ 1 + \alpha + \mu  - \nu } g ( \boldsymbol{x} )
                       -   \frac{ 1 + \alpha }{ 1 + \alpha + \mu  - \nu }  c ( \boldsymbol{x} ).$$
   According to the definition of $ \alpha $-weakly submodularity, we can easily prove that $ h_\alpha ( \boldsymbol{x^{*}}) $ also satisfies $ \alpha $-weakly submodularity.
   Using the same proof method applied in Sect. \ref{sec3}, we can obtain$ h_\alpha ( \boldsymbol{x^{*}} ) $ falls in $ [ m, km / \alpha ] $.  The memory is $ O ( k \log ( k / \alpha ) / \varepsilon)$.

\section{Conclusions}
\label{sec4}

In this paper, we present the development of combinatorial approximation algorithms designed to address the maximization problem of monotone non-negative submodular function minus a nonnegative linear function in integer lattice.
The problem is tackled under a cardinality constraint and by establishing an appropriate threshold, we start off with the creation of a bicriteria streaming algorithm. Running parallel to this, we utilize the lattice binary search algorithm as a subalgorithm.
Further into the subject matter, we explore scenarios where the first function of the objective funvtion is
$ \alpha  $-weakly submodular and a nonsubmodular function. For these situations, we design a combinatorial bicriteria streaming algorithm.
Simultaneously, we present comprehensive analyses for both models. Features of these models include that the overall memory is $ O ( k \log k / \varepsilon )$ while the query number per is  $ O(\log K)$.

\section*{Acknowledgements}
The first author  is supported by Natural Science Foundation of Shandong Province (No. ZR2022MA034) and National Natural Science Foundation of China (No. 12301417).
The second author is supported by National Natural Science Foundation of China (No. 11871081).
The fourth author is  supported by National Natural Science Foundation of China (Nos. 12131003, 12001025).





\end{document}